\documentclass[12pt]{amsart} 
\usepackage{amssymb} 
\usepackage[breaklinks]{hyperref} 
\usepackage{natbib} 
\setcitestyle{aysep={}} 
\RequirePackage{color,graphicx} 
\newcommand{\HH}{\mathcal{H}}

\newcommand{\RR}{\mathbb{R}}

\newcommand{\UU}{\mathcal{U}}

\newcommand{\kaon}{K^0}
\newcommand{\antikaon}{\bar{K}^0}

\newcommand\restr[2]{{
  \left.\kern-\nulldelimiterspace 
  #1 
  \vphantom{\big|} 
  \right|_{#2} 
  }}

\newcommand{\Inn}[2]{\langle #1, #2 \rangle}

\makeatletter
\let\uppercasenonmath\@gobble
\makeatother

\newtheorem{fact}{Fact}

\newtheorem{prop}{Proposition}
\newtheorem*{coro}{Corollary}
\theoremstyle{definition}
\newtheorem*{defi}{Definition}

\title{Three Merry Roads to T-violation}
\author[Bryan W. Roberts]{Bryan W. Roberts\\ University of Southern California \\ \href{http://www.usc.edu/bryanroberts}{www.usc.edu/bryanroberts} \\ \\ \today}

\begin{document}
\maketitle
\begin{abstract}
This paper is a tour of how the laws of nature can distinguish between the past and the future, or be $T$-\emph{violating}. I argue that, in terms of basic analytic arguments, there are really just three approaches currently being explored. I show how each is characterized by a symmetry principle, which provides a template for detecting $T$-violating laws even without knowing the laws of physics themselves. Each approach is illustrated with an example, and the prospects of each are considered in extensions of particle physics beyond the standard model.
\end{abstract}

\section{Introduction}\label{introduction}

Unlike thermal physics, the physics of fundamental particles does not normally distinguish between the past and the future. For example, most classical mechanical systems never do. This dogma once ran so deep that, even after the shocking discovery of \citet{Wu1956a} that parity or ``mirror symmetry'' is violated, it remained difficult to imagine the violation of temporal symmetry. Many simply considered it to be an unavoidable aspect of quantum field systems, because of the great simplification it provided in the description of weakly interacting particles\footnote{Cf. \citet{Weinbe1958a} and \citet{Lederman1956K_Ldiscovery}. As James Cronin colorfully put it: ``It just seemed evident that CP symmetry should hold. People are very thick-skulled. We all are. Even though parity had been overthrown a few years before, one was quite confident about $CP$ symmetry'' \citep{CroninGreenw1982a}. In the presence of $CPT$-invariance, $CP$ symmetry is equivalent to $T$ symmetry.}.

Since then, a great deal of evidence has been accumulated showing that, contrary to the early views of particle physicists, fundamental physics can be $T$-\emph{violating} --- it \emph{does} distinguish between the past and the future! I do not wish to retell that story here. There are many sources\footnote{For a book-length overview, try \citet{kabir1968,sachs1987,kleinknecht2003cp,sozzi-CPviolation,BigiSanda2009CPviolation}.}, which are really much better than me, that will explain to you all about the gritty and ingenious detections of $T$-violating interactions, the deep and beautiful theory underlying them, and how we can expect that theory to develop from here.

At this conference, I would like to attempt a different project, which is to draw out the basic analytic arguments underlying the various approaches to $T$-violation. I would like to cast these arguments into their bare skeletal form; to think about what makes them alike and distinct; and to ask how they may fare as particle physics is extended beyond what we know today. In sum, this will be a cheerful tour -- from a birds eye view, if you like -- of the existing roads to $T$-violation.

There are, I think, two main benefits to this abstract perspective. The first is to show that there are really only three distinct roads to $T$-violation from where we stand today. Each one is characterized by a symmetry principle, and each is a deductive consequence of quantum mechanics and quantum field theory. The second benefit of the abstract perspective is that it illustrates the powerful generality of our evidence for $T$-violation. We will see in particular that these approaches allow us to test whether the laws of physics are $T$-violating, \emph{even when we don't know what the correct laws of physics are!} Here is a summary of the three approaches to $T$-violation.

	\begin{enumerate}
		\item \emph{$T$-Violation by Curie's Principle}. Pierre Curie declared that there is never an asymmetric effect without an asymmetric cause. This idea, together with the so-called $CPT$ theorem, provided the road to the very first detection of $T$-violation in the 20th century.
		\item \emph{$T$-Violation by Kabir's Principle}. Pasha Kabir pointed that, whenever the probability of an ordinary particle decay $A\rightarrow B$ differs from that of the time-reversed decay $B^\prime\rightarrow A^\prime$, then we have $T$-violation. This provides a second road.
		\item \emph{$T$-Violation by Wigner's Principle}. Certain kinds of matter, such as an elementary electric dipole, turn out to be $T$-violating because they have an appropriate non-degenerate energy state\footnote{An energy eigenstate is degenerate if there exists an orthogonal eigenstate with the same eigenvalue. I will discuss this property in more detail below.}. This provides the final road, although it has not yet led to a successful detection of $T$-violation.
	\end{enumerate}

In the next three sections, I will explain each of these three roads to $T$-violation. Some of these roads are very exciting and surprising, especially if you have not travelled down them before, and I will try to keep things light-hearted for the newcomer. My explanations will begin with a somewhat abstract formulation of an analytic principle, followed by an illustration how it provides a way to test for $T$-violation, and then an elementary mathematical treatment. I'll end each section with a little discussion about the prospects for extensions of particle physics beyond the standard model, and in particular extensions in which the dynamical laws are not unitary.

Let's start at the beginning.

\section{$T$-violation by Curie's Principle}\label{sec:T-violation-and-Curies-principle}

The first evidence that the laws governing weakly interacting systems are $T$-violating was produced, rather incredibly, in the mid-1960's. This was before the standard model was formulated. It was before a complete understanding of weak interactions. I think it's fair to say that we had little knowledge of the correct laws describing these systems whatsoever, if one takes ``the laws'' to be given by a Lagrangian or Hamiltonian. So how could we know the laws are $T$-violating? It was through a clever principle first pointed out by the great French physicist Pierre Curie, and adopted by James Cronin and Val Fitch in their surprising discovery. Here is that story.

	\subsection{Curie's principle}

In 1894, Pierre Curie argued that physicists really ought to be more like crystallographers, in treating certain symmetry principles like explicit laws of nature. He emphasized one symmetry principle in particular, which has come to be known as \emph{Curie's principle:}
\begin{quote}
	When certain effects show a certain asymmetry, this asymmetry must be found in the causes which gave rise to them. \citep{curie1894a}
\end{quote}

To begin, we'll need to sharpen the statement of Curie's Principle, by replacing the language of ``cause'' and ``effect'' with something more precise. An obvious choice in particle physics is to take an ``effect'' to be a quantum state. What then is a cause? A natural answer is: the \emph{other} states in the trajectory (e.g. the states that came before), together with the law describing how those states dynamically evolve. So, Curie's principle can be more clearly formulated:

\begin{quote}
	If a quantum state fails to have a linear symmetry, then that asymmetry must also be found in either the initial state, or else in the dynamical laws.
\end{quote}
This is a common interpretation of Curie's principle\footnote{C.f. \citep{earman2004curie}, \citep[\S 9.2.4]{mittelstaedt_weingartner_laws}.}. In fact it can be sharpened even more, and we will do so shortly. But first let's now see how it appears in the surprising discovery of Cronin and Fitch.

	\subsection{Application to $CP$-violation}

The Cronin and Fitch discovery of $T$-violation really goes back to an incredible work by \citet{gell-mann-pais1955}, which among other things introduces a version of Curie's Principle. They did not refer to it in this way, but I think you will see that the principle is unmistakably Curie's. Let's start with the example of \emph{charge conjugation} (CC) symmetry, which has the effect of transforming particles into their antiparticles and vice versa. Suppose we have two particle states $\theta_1$ and $\theta_2$; their interpretation is not important for this point\footnote{Gell-Mann and Pais used $\theta^0_1$ and $\theta^0_2$ to refer to what we know call the neutral kaon states $K_1$ and $K_2$, discussed in Footnote \ref{fn:long-short-life-kaons} below.}. And suppose the state $\theta_1$ is ``even'' under charge conjugation, in that $C\theta_1=\theta_1$, while the state $\theta_2$ is ``odd,'' in that $C\theta_2 = -\theta_2$. Then, Gell-Mann and Pais observed,
\begin{quote}
	according to the postulate of rigorous CC invariance, the quantum number $C$ is conserved in the decay; the $\theta_1^0$ must go into a state that is even under charge conjugation, while the $\theta_2^0$ must go into one that is odd. \citep[p.1389]{gell-mann-pais1955}.
\end{quote}
Given $C$-symmetric laws, a $C$-symmetric state must evolve to another $C$-symmetric state. Or, reformulating this claim in another equivalent form: if a $C$-symmetric state evolves to a $C$-asymmetric state, \emph{then the laws themselves must be $C$-violating}. That's a neat way to test for symmetry violation. And it's a simple application of Curie's Principle.

Although Gell-Mann and Pais were discussing $C$-symmetry, the same reasoning applies to any linear symmetry whatsoever. In particular, it applies to $CP$-symmetry, which is the combined application of charge conjugation with the parity ($P$) or ``mirror flip'' transformation. Cronin later wrote that the Gell-Mann and Pais article ``sends shivers up and down your spine, especially when you find you understand it,'' pointing out that it suggests a statement that is clearly an application of Curie's Principle (although Cronin does not call it that):
\begin{quote}
	You can push this a little bit further and see how CP symmetry comes in. The fact that CP is odd for a long-lived $K$ meson means that $K_L$ could not decay into a $\pi^+$ and a $\pi^-$. If it does --- and that was our observation --- then there is something wrong with the assumption that the CP quantum number is conserved in the decay. \citep[p.41]{CroninGreenw1982a}
\end{quote}
Here is that reasoning in a little more detail. When you create a beam of neutral $K$ mesons or ``kaons,'' the long-lived state $K_L$ is all that's left after the beam has traveled a few meters\footnote{\label{fn:long-short-life-kaons} The study of strong interactions had led to the identification of kaon particle and antiparticle states $K_0$ and $\bar{K}_0$ that are eigenstates of a degree of freedom called \emph{strangeness}. When testing for $CP$-violation, it is easier to study the superpositions $K_1=(K^0+\bar{K}^0)/\sqrt{2}$ and $K_2=(K^0-\bar{K})/\sqrt{2}$, since the lifetime of the latter is orders of magnitude longer. At the time, $K_2$ was identified as the ``long-life kaon state $K_L$.''}. This long-lived state had been discovered eight years earlier in the same laboratory by \citet{Lederman1956K_Ldiscovery}. And it was known that $K_L$ is \emph{not} invariant under the $CP$ transformation, whereas a two pion state $\pi^+\pi^-$ \emph{is} invariant under $CP$. The observation of such the asymmetric decay $K_L\rightarrow\pi^+\pi^-$, Cronin points out, could only be the result of a $CP$-violating law. That's just Curie's Principle.
\begin{figure}[tb]\begin{center}
\includegraphics[width=0.5\textwidth]{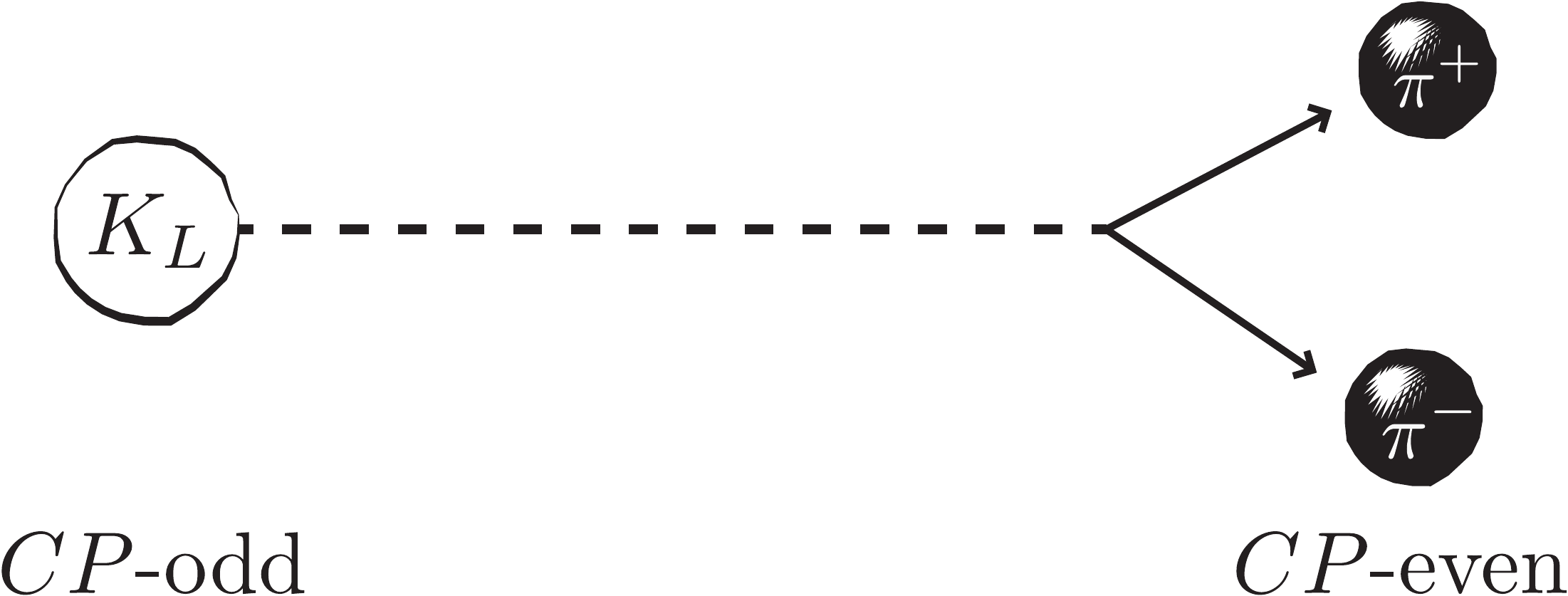}
\caption{The $K_L\rightarrow\pi^+\pi^-$ decay. By Curie's Principle, this asymmetry between an initial state and a final state implies an asymmetry in the laws.}\label{fig:kaon1}
\end{center}\end{figure}

The Cronin and Fitch experiment of 1964 involved firing a $K_L$ beam into a spark chamber at the Brookhaven National Laboratory, and taking photographs of thousands of particle decay events occurring over the course of about $10^{-10}$ seconds. Their ``Eureka moment'' was somewhat of a delayed reaction, as they labored for months analyzing all the particle events that they had photographed\footnote{The history of this discovery is recalled in a charming lecture given by Cronin at the University of Chicago and transcribed by \citep{CroninGreenw1982a}.}. But when the analysis was complete, they found that some of the $K_L$ kaons decayed into a pair of pions, $K_L\rightarrow\pi^+\pi^-$. This decay event was rare, occurring in only about one in every 500 of the recorded decays, but it was nonetheless unmistakable. The conclusion, by a simple application of Curie's Principle, was that the laws must be $CP$-violating. Cronin and Fitch told Abraham Pais about their exciting discovery over coffee at Brookhaven. Pais later wrote about their conversation that, ``[a]fter they left I had another coffee. I was shaken by the news'' \citep{pais1990a}. Cronin and Fitch were awarded the 1980 Nobel Prize for their discovery.

Of course, there were \emph{many} deep insights that led to the discovery of $CP$-violation. They included the discovery of the strangeness degree of freedom, the prediction of kaon-antikaon oscillations, the discovery of the long-lived $K_L$ state, the understanding of kaon regeneration, and many other things. But I hope to have shown here that, in skeletal form, the first argument for $CP$-violation is really a simple application of Curie's Principle.

	\subsection{The conclusion of $T$-violation}

The final step to the conclusion of $T$-violation now follows from the so-called $CPT$-theorem. Virtually all known laws of physics are invariant under the combined transformation of charge-conjugation ($C$), parity ($P$), and time reversal ($T$). Of course, the Hamiltonian governing the decay of the neutral kaon was not known in 1964, and so we could hardly just check whether it's $CPT$-invariant. But there was a theorem to assure us that, at least for quantum theory as we know it --- describable in terms of local (Wightman) fields, and a unitary representation of the Poincar\'e group --- the laws must be invariant under $CPT$. This result is called the $CPT$ \emph{theorem}, first proved in this form by \citet{jost1957cpt}, although arguments of a similar character were given by many others\footnote{For example, \citet{pauli1955} derives $CPT$ invariance as a corollary to the spin-statistics theorem, and \citet{BorchersYngvason2001a} derived it from the Haag axioms.}. And it straightforwardly implies that if $CP$ is violated, $T$ must be violated as well\footnote{$CPT$-invariance says that $(CPT)H=H(CPT)$, and thus that $CP(THT^{-1})=(H)CP$. So, if we have time reversal invariance, then the left-hand term $THT^{-1}$ gets set to $H$, and we immediately have $CP$-invariance, $CP(H)=(H)CP$. Equivalently, if $CP$ invariance fails, then so does time reversal invariance.}.

Thus, insofar as the premises of the $CPT$ theorem apply to our world, the Cronin and Fitch application of Curie's principle provides immediate evidence for $T$-violation as well.

\subsection{Mathematical underpinning}\label{subs:math-underpinning}

The statement of Curie's principle described above is not just a helpful folk-theorem. It can be given precise mathematical expression. Let me now try to make the mathematics more clear. I'll begin with a very simple mathematical statement of Curie's Principle in terms of unitary evolution, and then show how it can be carried over to scattering theory.

To begin, recall what it means for a law to be invariant under a linear symmetry transformation $R$.
\begin{defi}[invariance of a law]
A law of physics is \emph{invariant} under a linear transformation $R$ if whenever $\psi(t)$ is an allowed trajectory according to the law, then so is $R\psi(t)$.
\end{defi}
In the standard model of particle physics, interactions are assumed to evolve unitarily over time, by way of a continuous unitary group $\UU_t=e^{-itH}$, where $H$ is the Hamiltonian generator of $\UU_t$. In this context, the above definition of invariance is equivalent to
\begin{equation*}
	[H,R]=0
\end{equation*}
where $H$ again is the Hamiltonian and $R$ is linear. In these terms, we can give a first formulation of Curie's Principle as follows\footnote{A version of this fact was pointed out by \citet[Prop. 2]{earman2004curie}.}.

\begin{fact}[Unitary Curie Principle]\label{fact:unitary-curie}
	Let $\UU_t=e^{-itH}$ be a continuous unitary group on a Hilbert space $\HH$, and $R:\HH\rightarrow\HH$ be a linear bijection. Let $\psi_i\in\HH$ (an ``initial state'') and $\psi_f = e^{-itH}\psi_i$ (a ``final state'') for some $t\in\RR$. If either
	\begin{enumerate}
		\item (initial but not final) $R\psi_i = \psi_i$ but $R\psi_f \neq \psi_f$, or 
		\item (final but not initial) $R\psi_f = \psi_f$ but $R\psi_i \neq \psi_i$,
	\end{enumerate}
	then,
	\begin{enumerate}\setcounter{enumi}{2}   
		\item ($R$-violation) $[R,H]\neq0$.
	\end{enumerate}
\end{fact}
\begin{proof}
Suppose that $[R,H]=0$, and hence (since $R$ is linear) that $[R, e^{-itH}]=0$. Then $R\psi_i=\psi_i$ if and only if $R\psi_f = Re^{-itH}\psi_i = e^{-itH}R\psi_i = e^{-itH}\psi_i = \psi_f$. 
\end{proof}
This, again, is just a helpful first formulation. We have not yet arrived at a principle that is appropriate for the description of $CP$-violation. The claim of Cronin and Fitch was that in a neutral kaon scattering event, \emph{there is a particular decay mode} $K_L\rightarrow\pi^+\pi^-$ that occurs only if the laws are $CP$-violating $[CP,H]\neq0$. We have not yet given a rigorous formulation of \emph{that} application of Curie's Principle.

To get there, we first observe that it is enough for $CP$ to fail to commute with the $S$-matrix, $[CP,S]\neq0$. For, if a symmetry $R$ commutes with the ``free'' part of the Hamiltonian $[R,H_0]=0$ (which is true of most familiar symmetries, including $CP$), then by the definition of the $S$-matrix\footnote{\label{fn:s-violation-h-violation} One easy way to see this is to just look at the explicit Dyson expression of the $S$-matrix, \begin{equation}
	S = \mathcal{T}\exp\left(-i\int_{-\infty}^\infty dt V_I(t)\right),
\end{equation}
where $V_I$ is the interacting part of the Hamiltonian $H=H_0+V_I$, and $\mathcal{T}$ is the time-ordered multiplication operator \citep[p.73]{sakurai1994}. If $H=H_0+V_I$, then $[R,H_0]=0$ and $[R,H]=0$ implies that $[R,V_I]=[R,H-H_0]=[R,H]-[R,H_0]=0$. Thus, since $R$ is linear, we can pass it through the integral above to get that $RSR^{-1}=S$.},
\begin{equation*}
	[R,S]\neq0 \text{ only if } [R,H]\neq0.
\end{equation*}
Thus, by showing that the scattering matrix is $CP$-violating, one equally shows that the unitary dynamics are $CP$-violating as well. With this in mind, we can now state Curie's Principle in a form that is more appropriate for scattering theory.

\begin{fact}[Scattering Curie Principle]\label{fact:scattering-curie}
	Let $S$ be a scattering matrix, and $R:\HH\rightarrow\HH$ be a unitary bijection. If there exists any decay channel $\psi^{in}\rightarrow\psi^{out}$ such that either,
	\begin{enumerate}
		\item (in but not out) $R\psi^{in} = \psi^{in}$ but $R\psi^{out} =-\psi^{out}$, or 
		\item (out but not in) $R\psi^{out} = \psi^{out}$ but $R\psi^{in}=-\psi^{in}$,
	\end{enumerate}
then,
	\begin{enumerate}\setcounter{enumi}{2}   
		\item $[R,S]\neq0$.
	\end{enumerate}
Moreover, if $\UU_t=e^{-it(H_0+V)}$ is the associated unitary group, and if $R$ commutes with the free component $H_0$ of the Hamiltonian $H=H_0+V$, then ($R$-violation) $[R,H]\neq0$.
\end{fact}
\begin{proof}
We prove the contrapositive; suppose that $[R,S]=0$. Since $R$ is unitary, $\Inn{\psi^{out}}{S\psi^{in}} = \Inn{R\psi^{out}}{RS\psi^{in}} = \Inn{R\psi^{out}}{SR\psi^{in}}$. So, if either the ``in but not out'' or the ``out but not in'' conditions hold, then,
\begin{equation*}
	\Inn{\psi^{out}}{S\psi^{in}} = \Inn{R\psi^{out}}{SR\psi^{in}} = -\Inn{\psi^{out}}{S\psi^{in}}.
\end{equation*}
Hence, $\Inn{\psi^{out}}{S\psi^{in}}=0$, which means that there can be no decay channel $\psi^{in}\rightarrow\psi^{out}$. Finally, we note that if $[R,H_0]=0$, then and $[R,S]\neq0$ implies that $[R,H]\neq0$ by the definition of the $S$-matrix.
\end{proof}
This, finally, is the precise mathematical statement of Curie's Principle that was applied by Cronin and Fitch. Taking $\psi^{in}=K_L$ and $\psi^{out}=\pi^+\pi^-$, they discovered a scattering event $\psi^{in}\rightarrow\psi^{out}$ that satisfies  ``out but not in'' for the transformation $R=CP$. It follows that the laws are $CP$-violating. And given $CPT$ invariance, it follows that they are $T$-violating as well.

	\subsection{Advantages and limitations}\label{subs:curie-advant-limit}

An obvious advantage of this approach to $T$-violation is that you don't have to know the laws to know that they are $T$-violating. At the time of its discovery in 1964, many of the structures appearing in the modern laws of neutral kaon decay were absent: there were no $W$ or $Z$ bosons, no Kobayashi-Maskawa matrix, and certainly no standard model of particle physics. All that came later. Nevertheless, Curie's Principle provided a surprisingly simple test that the laws are $T$-violating.

A more subtle advantage is that, as a test for $CP$ violation, Curie's Principle will likely continue hold water in non-unitary extensions of quantum theory\footnote{\citet{ashtekar2013threeroads} has formulated a version of Curie's principle that applies much more generally than the one I have stated here.}. Although unitary evolution is assumed in some of the background definitions, nothing about the argument from Curie's Principle requires the evolution be unitary. For example, the ``scattering version'' of Curie's principle in no way depends on the unitarity of the $S$-matrix; indeed, the conclusion that $[R,S]\neq0$ holds when $S$ is any Hilbert space operator whatsoever that connects $\psi^{in}$ and $\psi^{out}$ states. In this sense, the argument from Curie's principle is very general indeed.

The limitation is that it is an indirect test for $T$-violation, and one that we might not trust as we attempt to extend particle physics beyond the standard model. In particular, the reliance on the $CPT$ theorem is troubling. It is not implausible that $CPT$ invariance could fail as particle physics is extended beyond the standard model. For example, we might wish to consider a representation of the Poincar\'e group that is not completely unitary. In such cases, the $CPT$ theorem can fail, and thus so would the link between $CP$-violation and $T$-violation. It would be preferable to have a direct test of $T$-violation instead.

One might respond to this concern by trying to apply Curie's Principle directly to the case of $T$-violation. Unfortunately, that doesn't work. Recall that the statement of Curie's Principle above assumed the symmetry transformation was linear. This turns out to be a crucial assumption; Curie's Principle fails badly for antilinear symmetries like time reversal\footnote{For a discussion of this failure of Curie's principle, see \citet{roberts-2013a}.}. So, this road to $T$-violation is \emph{essentially} indirect. One can check directly for $CP$ violation, but only recover $T$-violation by the $CPT$ theorem. A direct test of $T$-violation will have to follow a completely different argument. That is the topic of the next section.

\section{$T$-Violation by Kabir's Principle}\label{sec:direct-t-violation}

New progress has recently been made in the understanding of $T$-violation. We now have evidence that appears to be much more direct. The first such evidence came with an experiment by \citet{AngelopoulosEtAl1998}, performed at the CPLEAR particle detector at CERN. Like the original $T$-violation experiment, this discovery involved the decay of neutral kaons. Things got even better when, just a few months ago now, yet another direct detection of $T$-violation was announced by the BaBar collaboration at Stanford \citep{LeesPoirea2012a}. This experiment involved the decay of a different particle, the neutral $B$ meson. It's an exciting time for the study of $T$-violation! But for our purposes, what's special about these new results is that the argument underlying them is completely different from that adopted by Cronin and Fitch. No application of Curie's Principle is needed.

What I would like to point out is that both recent detections of $T$-violation hinge on another symmetry principle. Let me call it \emph{Kabir's Principle}, since it was pointed out in an influential pair of papers by \citet{kabir1968nature,PhysRevD.2.540}. Unlike the Curie Principle approach to symmetry violation, this one is really built to handle antilinear transformations like time reversal. Here is how it works.

	\subsection{Kabir's Principle}

To begin, let me summarize the simple idea behind Kabir's Principle somewhat roughly.
\begin{quote}
	If a transition $\psi^{in}\rightarrow\psi^{out}$ occurs with different probability than the time-reversed transition $T\psi^{out}\rightarrow T\psi^{in}$, then the laws describing those transitions must be $T$-violating.
\end{quote}
This suggests a straightforward technique for checking whether or not an interaction is governed by $T$-violating laws. We set up a detector to check how often a particle decay $\psi_i\rightarrow\psi_f$ occurs (called its \emph{branching ratio}), and compare it to how often a the decay $T\psi_f\rightarrow T\psi_i$ occurs. Easier said than done, naturally. But if one occurs more often than the other, then Kabir's Principle says we have direct evidence of $T$-violation.

In the next subsection, I will sketch briefly how such a procedure was first carried out at CERN. I'll then discuss the precise mathematical formulation of Kabir's Principle.

\subsection{Application to $T$-violation}

The first direct detection of $T$-violation involved the decay of our friend the neutral kaon. So, let's return to the strangeness eigenstates $K^0$ and $\bar{K}^0$, which have strangeness eigenvalues $\pm1$, respectively. It is generally thought that, if strong interactions were all that governs the behavior of these states, then strangeness would be conserved. So, by the kind of arguments discussed above, you could never have a particle decay like $K^0\rightarrow\bar{K}^0$ with only strong interactions, because these states have different values of strangeness. However -- and this is another thing pointed out in the remarkable article by \citet{gell-mann-pais1955} -- when weak interactions are in play, there is no reason not to entertain decay channels that fail to conserve strangeness.

In fact, in the presence weak interactions, it makes sense to consider both $K^0\rightarrow\bar{K}^0$ and $\bar{K}^0\rightarrow K^0$ as possible decay modes. These particles could in principle bounce back and forth between each other, $K^0\rightleftarrows\bar{K}^0$, by a phenomenon called \emph{kaon oscillation}. This is a very exotic property, which applies to only a few known particles (one of them being the $B$ meson), and it is part of what makes neutral kaons so wonderfully weird.

Now, a convenient thing about oscillations between $K^0$ and $\bar{K}^0$ is that they are very easy to time reverse. In particular we can always set the phases\footnote{There is a great deal of freedom in choosing the phase conventions for the discrete transformations of $K^0$; see \citet[\S 9]{sachs1987} for a discussion.} so that,
\begin{align*}
	T\kaon = \kaon, \;\;\;\; T\antikaon = \antikaon.
\end{align*}
This allows us to apply Kabir's Principle in a particularly simple form: if we observe $K^0\rightarrow\bar{K}^0$ to occur with a different probability than $\bar{K}^0\rightarrow K^0$, then we have direct evidence for $T$-violation! This is precisely what was found at the CPLEAR detector, in showing that there is ``time-reversal symmetry violation through a comparison of the probabilities of $\bar{K}^0$ transforming into $K^0$ and $K^0$ into $\bar{K}^0$'' \citep{AngelopoulosEtAl1998}.

\begin{figure}[tb]\begin{center}
\includegraphics[width=0.9\textwidth]{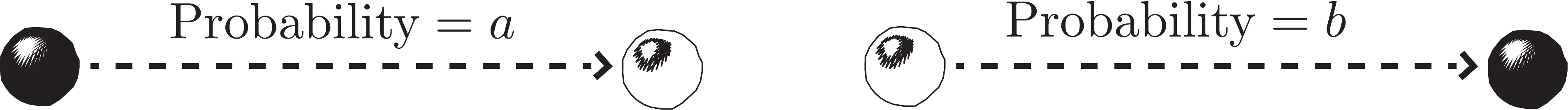}
\caption{Application of Kabir's Principle. If the decay $K^0\rightarrow\bar{K}^0$ happens more often than the time-reversed decay $\bar{K}^0\rightarrow K^0$, then the interaction is $T$-violating.}\label{fig:kaon2}
\end{center}\end{figure}

At this level of abstraction, it is the very same strategy that was used in the Stanford $T$-violation experiment with $B$ mesons. It turns out that neutral $B$ mesons can also oscillate between two states, $B^0\rightleftarrows B_-$. \citet{BernabMartin2012a} pointed out that if these transitions were to occur with different probabilities, then we would have $T$-violation. And this is just what was recently detected by \citet{LeesPoirea2012a} at Stanford. Thus, both the Stanford detection and the original CPLEAR detection $T$-violation were made possible by the abandonment of Curie's Principle, in favor of the more the more direct principle of Kabir.

\subsection{Mathematical Underpinning}

As with Curie's Principle, Kabir's Principle has a rigorous mathematical underpinning. But before getting to that, it's important to note the special way that unitary operators like the $\UU_t=e^{-itH}$ and the $S$-matrix transform under time reversal. The point where many get stuck is on the fact that $T$ is \emph{antiunitary}. This means that it conjugates the amplitudes, $\Inn{T\psi}{T\phi}=\Inn{\psi}{\phi}^*.$ It also means that it is \emph{antilinear}, in that it conjugates any complex number that we pass it over:
\begin{equation*}
	T(a\psi+b\phi)=a^*T\psi+b^*T\phi.
\end{equation*}
As a consequence, the condition of time reversal invariance that $[T,H]=0$ does \emph{not} imply that the unitary operator $\UU_t=e^{-itH}$ commutes with $T$. Instead, the complex constant picks up a negative sign. That is, for time reversal invariant systems,
\begin{equation*}
	T\UU_tT^{-1} = e^{-(-itTHT^{-1})}=e^{itH}=\UU_{-t}=\UU_t^{-1}.
\end{equation*}
Similarly, a unitary $S$-matrix describes a time-reversal invariant system if and only if $TST^{-1}=S^{-1}$.

We can now formulate a mathematical statement of Kabir's Principle. Note that, as discussed in Section \ref{subs:math-underpinning}, the failure of the $S$-matrix to be time reversal invariant ($TST^{-1}\neq S^{-1}$) implies $T$-violation in the ordinary sense ($T\UU_tT^{-1}\neq \UU_t^{-1}$).
\begin{fact}[Kabir's Principle]\label{fact:reversal-principle}
	Let $S$ be a unitary operator (the $S$-matrix) on a Hilbert space $\HH$, and let $T:\HH\rightarrow\HH$ be an antiunitary bijection. If,
	\begin{enumerate}
		\item (unequal amplitudes) $\Inn{\psi^{in}}{S\psi^{out}} \neq \Inn{T\psi^{out}}{ST\psi^{in}}$,
	\end{enumerate}
then,
	\begin{enumerate}\setcounter{enumi}{1}   
		\item ($T$-violation) $TST^{-1}\neq S^{-1}$.
	\end{enumerate}
\end{fact}
\begin{proof}
	We argue the contrapositive. Assume $TST^{-1}=S^{-1}$. $T$ is antiunitary, so $\Inn{\psi^{out}}{S\psi^{in}} = \Inn{T\psi^{out}}{TS\psi^{in}}^*$. Thus, since $TS=S^{-1}T$ by assumption,
\begin{align}\label{eq:kabir-proof}
	\begin{split}
		\Inn{\psi^{out}}{S\psi^{in}} & = \Inn{T\psi^{out}}{S^{-1}T\psi^{in}}^*\\
									 & = \Inn{S^{-1}T\psi^{in}}{T\psi^{out}}\\
									 & = \Inn{T\psi^{in}}{ST\psi^{out}},
	\end{split}
\end{align}
where the last equality follows from our claim that $S$ is unitary.
\end{proof}

	\subsection{Advantages and limitations}
	
Kabir's Principle, like that of Curie, provides a way to show the laws are $T$-violating without actually knowing much about the laws themselves. But even better, it does so without recourse to the $CPT$ theorem. In this sense, Kabir's Principle stands a better chance of remaining valid in $CPT$-violating extensions of the standard model.

A limitation is that, unlike the Curie's Principle approach, Kabir's Principle only seems to work when the dynamics is unitary. As in Section \ref{subs:curie-advant-limit}, suppose we consider some non-unitary extension of the standard model. That is, suppose our dynamics is not described by a (linear) unitary operator $\UU_t$. Then the above argument for Kabir's Principle fails in the final step\footnote{Nevertheless, \citet{ashtekar2013threeroads} has shown that a weaker condition than full unitary is still enough to establish Kabir's Principle; all that is needed is a condition of compatibility with an overlap map, without any linear structure involved.}.

Thus, although the Kabir Principle applied by \citet{AngelopoulosEtAl1998} and \citet{LeesPoirea2012a} has the advantage of providing a direct test, it is not clear that it is general enough to apply without modification in the context of a non-unitary extension of quantum theory.

\section{$T$-violation by Wigner's principle}\label{sec:T-invariance-and-Kramers-degeneracy}

I'd like to finish with one final road to $T$-violation. It is perhaps the least well-known of all the approaches. In simplest terms, this road involves the search for exotic new properties of matter. Let me begin with a toy model of how this can lead to $T$-violation. I'll then turn to the general reasoning underpinning this approach, and finally show how this reasoning has been applied (unsuccessfully so far) in empirical tests.

\subsection{A toy example} An electric dipole moment typically describes the displacement between two opposite charges, or within a distribution of charges. But suppose that, instead of describing a distribution of charges, we imagine an electric dipole moment as a property of just one elementary particle. This particle might be referred to as an ``elementary'' electric dipole moment.

The existence of such a property has been entertained, for example as a property of the neutron, although it has not yet been detected. Let $H_0$ be the Hamiltonian describing the particle in the absence of interactions; let $J$ represent its angular momentum; and let $E$ be an electromagnetic field. Then these ``elementary'' electric dipoles are sometimes\footnote{\citep[See][]{Khrip-Lam1997cp,DallRitz2013a}} characterized by the Hamiltonian,
\begin{equation*}
	H = H_0 + J\cdot E.
\end{equation*}
Since time reversal preserves the free Hamiltonian $H_0$ and the electric field $E,$ but reverses angular momentum $J$, this Hamiltonian is manifestly $T$-violating: $[T,H]\neq0$. Therefore, an elementary electric dipole of this kind would constitute a direct detection of $T$-violation. No need for Curie's Principle. No need for Kabir's Principle. No need for the $CPT$ theorem.

Like the $T$-violating $K_L\rightarrow\pi^+\pi^-$ and $K^0\rightleftarrows\bar{K}^0$ decays, there are general principles underpinning this example of $T$-violation, too. In this case, they stem from the relationship between $T$-invariance and the degeneracy of the energy spectrum. The relevant relationship can be summarized as follows.

\subsection{Wigner's Principle}\label{subs:non-degeneracy-principle}

A system is called \emph{degenerate} if its Hamiltonian has distinct energy states with the same energy eigenvalue. An intuitive example of a degenerate system is the free particle on a string: the particle can either move to the left, or to the right, and have the same kinetic energy either way. When there are multiple distinct eigenstates with the same eigenvalue, those eigenstates are called \emph{degenerate states}.  \citet{kramers1930degeneracy} showed that an odd number of electrons can be expected to have a degenerate energy spectrum, and for this his name remains attached to that effect: Kramers Degeneracy\footnote{The reason people were interested in the first place, it seems, is that degeneracy was a key part of knowing how to studying very low temperature phenomena using paramagnetic salts \citep{Klein_Kramer}.}. But it was \citet{wigner1932kramer} showed the much deeper relationship between degeneracy and time reversal invariance.

For the purposes of understanding $T$-violation, the relevant relationship can be summarized as follows.
\begin{fact}[Wigner's Principle]\label{fact:non-degeneracy-principle}
	If there is an eigenstate of the Hamiltonian such that: (1) that state is non-degenerate, and (2) time reversal maps that state to a different ray, then we have $T$-violation, in that $[T,H]\neq0$.
\end{fact}
We will see shortly how this fact has a simple proof deriving from the work of Wigner. But first, let me point out how it can be used to provide evidence for $T$-violation.

	\subsection{Application to $T$-violation}

We observed above that an appropriately weird Hamiltonian can provide an explicit and direct example of $T$-violation. The properties that these systems tend to share, it turns out, are just the properties described by Wigner's principle above. There are various examples that one could study here to illustrate. But let me spare the reader and give just one that is rather important, the elementary electric dipole moment.

The thing that is not obvious about the elementary electric dipole moment is that it satisfies part (1) of Wigner's principle whenever part (2) is satisfied. That is, time reversal acts non-trivially on all the energy eigenstates $\psi$ of the Hamiltonian ($T\psi\neq e^{i\theta}\psi$) that are non-degenerate. So, if for example one makes the common assumption that the stable ground state $\psi$ of an elementary particle is non-degenerate, then for an elementary electric dipole we also have that $T\psi\neq e^{i\theta}\psi$. It follows by Wigner's principle that this system is $T$-violating.

To begin, let me briefly introduce the elementary electric dipole moment\footnote{For more details, see \citet[\S 13.3]{ballentine1998}, \citet[\S XXI.31]{messiah1999}, or \citet[\S 4.2]{sachs1987}.}. It can be characterized as a system with the following three properties.
	\begin{enumerate}
		\item \emph{(Permanence)} There is an observable $D$ representing the dipole moment that is ``permanent'', in that $\Inn{\psi}{D\psi}=a>0$ for every eigenvector $\psi$ of the Hamiltonian $H$. Since this $\psi(t)$ does not change over time except for a phase factor, permanence means that $\Inn{\psi}{D\psi}=a$ has the same non-zero value for all times $t$, whence its name.
		\item \emph{(Isotropic Dynamics)} Assuming that we have elementary particle, its simplest interactions are assumed to be isotropic, in that time evolution commutes with all rotations, $[e^{-itH}, R_\theta]=0$. Note that if $J$ is the ``angular momentum'' observable that generates the rotation $R_\theta=e^{i\theta J}$, then this is equivalent to the statement that $[H,J]=0$.
		\item \emph{(Time Reversal Properties)} Time reversal, as always, is an antiunitary operator. It has no effect on the electric dipole observable ($TDT^{-1}=D$) when viewed as a function of position. But it does reverse the sign of angular momentum $(TJT^{-1}=-J)$, since spinning things spin in the opposite orientation when their motion is reversed.
	\end{enumerate}

	A system with these three properties turns out to satisfy condition (1) of Wigner's principle, that $T\psi\neq e^{i\theta}\psi$ for some eigenvector $\psi$ of $H$, whenever that $\psi$ is non-degenerate. To see why, assume (for reductio) that there is a non-degenerate eigenvector $\psi$ of $H$ satisfying  $T\psi=e^{i\theta}\psi$. We will show that the assumption that the dipole moment is ``permanent'' then fails, contradicting our hypothesis.

	Since $[H,J]=0$, and since $\psi$ is a non-degenerate eigenvector of $H$, it follows\footnote{Check: $H(J\psi)=JH\psi=h(J\psi)$, so $J\psi$ is an eigenvector of $H$ with eigenvalue $h$; thus, by non-degeneracy, $J\psi=e^{i\theta}\psi$, and so $\psi$ is an eigenvector of $J$ with eigenvalue $e^{i\theta}$ (where since $J$ is self-adjoint, $e^{i\theta}=\pm1$).} that $\psi$ is an eigenvector of $J$. By the Wigner-Eckart Theorem\footnote{A special case of this theorem states that for any fixed eigenvector of angular momentum, the matrix elements of a vector observable are proportional to those of angular momentum. \citep[See][\S 7.2, esp. page 199]{ballentine1998}.}, each such eigenvector satisfies,
	\begin{equation}\label{eq:wigner-eckart-dipole}
		\Inn{\psi}{D\psi}=c\Inn{\psi}{J\psi}
	\end{equation}
	for some $c\in\RR$. Applying the antiunitary time reversal operator $T$ to vectors on both sides we get that $\Inn{T\psi}{TD\psi}^*=c\Inn{T\psi}{TJ\psi}^*$, and hence $\Inn{T\psi}{TD\psi}=c\Inn{T\psi}{TJ\psi}$. But $T$ commutes with $D$ and anticommutes with $J$, so this equation may be written,
	\begin{equation}\label{eq:wigner-eckart-dipole-reversed}
		\Inn{T\psi}{D(T\psi)} = -c\Inn{T\psi}{J(T\psi)}
	\end{equation}
	Finally, we have assumed (for reductio) that $T\psi=e^{i\theta}\psi$ for some $e^{i\theta}$. Applying this to Equation \eqref{eq:wigner-eckart-dipole-reversed}, we get,
	\begin{equation*}
				   \Inn{\psi}{D\psi} = - c\Inn{\psi}{J\psi}.
	\end{equation*}
Comparing this to Equation \eqref{eq:wigner-eckart-dipole}, we see that $\Inn{\psi}{D\psi}=-\Inn{\psi}{D\psi}$, and hence that $\Inn{\psi}{D\psi}=0$. This contradicts our hypothesis that $D$ is permanent.
	
So, if the elementary electric dipole has a non-degenerate energy energy eigenvector $\psi$, then $T\psi\neq e^{i\theta}\psi$. Wigner's Principle thus guarantees that it is a $T$-violating system. Constructing such a system is part of an active search for $T$-violation.

There are many interesting things to say about this research; for a book-length treatment, see \citet{Khrip-Lam1997cp}. All I would like to point out for now is that this approach to $T$-violation hinges on Wigner's principle, which is distinct from all the other approaches to $T$-violation discussed so far.

	\subsection{Mathematical Underpinning}

As suggested above, Fact \ref{fact:non-degeneracy-principle} basically arises out of Wigner's discovery of a connection between time reversal and degeneracy for systems with an odd number of fermions. Here is how that connection leads to a principle for understanding $T$-violation.

Wigner began by noticing a strange fact about two successive applications of the time reversal operator $T$. When applied to a system consisting of an odd number of electrons, it does not exactly bring an electron back to where we started. Instead, it adds a phase factor of $-1$. Only by applying time reversal twice more can we return an electron to its original vector state. This is a curious property indeed! But there is no getting around it. It is effectively forced on us by the definition of time reversal and of a spin-$1/2$ system \citep{roberts2012a}.

This led Wigner to the following argument that the electron always has a degenerate Hamiltonian\footnote{Wigner's assumption of a finite-dimensional Hilbert space can be relaxed, as generalizations exist for Hamiltonians with a continuous energy spectrum as well \citep{roberts2012a}.}.

\begin{prop}[Wigner]\label{prop:wigner}
Let $H$ be a self-adjoint operator on a finite-dimensional Hilbert space, which is not the zero operator. Let $T:\HH\rightarrow\HH$ be an antiunitary bijection. If
	\begin{enumerate}
		\item (electron condition) $T^2=-I$, and
		\item ($T$-invariance) $[T,H]=0$
	\end{enumerate}
then,
	\begin{enumerate}\setcounter{enumi}{2}
		\item (complete degeneracy) every eigenvector of $H$ admits an orthogonal eigenvector with the same eigenvalue.
	\end{enumerate}
\end{prop}

That's a fine argument for degeneracy, when we are confident about time reversal invariance. But what if we are interested in systems that are $T$-violating? No problem. We can just interpret Wigner's result in the following equivalent form.

\begin{coro}\label{coro:reverse-wigner}
Let $H$ be a self-adjoint operator on a finite-dimensional Hilbert space, which is not the zero operator. Let $T:\HH\rightarrow\HH$ be an antiunitary bijection. If
	\begin{enumerate}
		\item (electron condition) $T^2=-I$, and
		\item (non-degeneracy) there is an eigenvector of $H$ such that every eigenvector orthogonal to it has a different eigenvalue,
	\end{enumerate}
then,
	\begin{enumerate}\setcounter{enumi}{2}
		\item ($T$-violation) $[T,H]\neq0$.
	\end{enumerate}
\end{coro}
This means that Wigner's result is actually a toy strategy for testing $T$-violation in disguise! Suppose we discover an electron described by a non-degenerate Hamiltonian. Then we will have achieved a direct detection of $T$-violation.

There is a more general sort of reasoning at work here. It turns out that the $T^2=-I$ condition is stronger than is really needed to prove the result. The following generalization, which otherwise follows Wigner's basic argument, is available.
   
\begin{prop}[Wigner's Principle]\label{prop:reverse-wigner-generalized}
 Let $H$ be a self-adjoint operator on a finite-dimensional Hilbert space, which is not the zero operator. Let $T$ be an antiunitary bijection. If there exists an eigenvector $\psi$ of $H$ such that,
		\begin{enumerate}
			\item $T\psi\neq e^{i\theta}\psi$ for any complex unit $e^{i\theta}$, and
			\item every eigenvector orthogonal to $\psi$ has a different eigenvalue,
		\end{enumerate}
	then,
		\begin{enumerate}\setcounter{enumi}{2}
			\item ($T$-violation) $[T,H]\neq0$
		\end{enumerate}
\end{prop}
\begin{proof}
	We prove the contrapositive, by assuming (3) fails, and proving that there exists an vector for which either (1) or (2) fails as well. Let $H\psi=h\psi$ for some $h\neq0$ and some eigenvector $\psi$ of unit norm. Since $T$ is antiunitary, $T\psi$ will also have unit norm.
	
Suppose (3) fails, and hence that $[T,H]=0$. Then $H(T\psi)=TH\psi=h(T\psi)$. This means that if $\psi$ is any eigenvector of $H$ with eigenvalue $h$, then $T\psi$ is an eigenvector with the same eigenvalue. By the spectral theorem, the eigenvectors of $H$ form an orthonormal basis set. So, since $\psi$ and $T\psi$ are both unit eigenvectors, either $T\psi=e^{i\theta}\psi$ or $\Inn{T\psi}{\psi}=0$. The latter violates condition (2), and the former violates the condition (1). Therefore, either (1) or (2) must fail.
\end{proof}

What I am calling Wigner's Principle is thus a simple generalization of Wigner's insight into Kramers degeneracy. And it is this very principle that provides that basic analytic grounding for our final road to $T$-violation.

\subsection{Advantages and Limitations} Wigner's Principle provides a test for $T$-violation without appeal to any fancy phenomena like neutral kaon decay. Good old electromagnetic interactions are enough, if exotic properties of matter like an elementary electric dipole happen to exist. The criterion is very simple: if time reversal takes a non-degenerate energy eigenstate to a distinct ray, then we have $T$-violation.

A disadvantage is that it is harder to apply Wigner's Principle outside the context of standard quantum mechanics and quantum field theory. Degeneracy is a concept that finds its most natural home in quantum theory, and it is essential to the Wigner's Principle \citep[but for a discussion of its generalization, see][]{roberts2013}. The principle also requires us to know when a system admits an appropriate (non-degenerate) energy eigenstate, which may require more detailed knowledge about the Hamiltonian of a system than the other two roads to $T$-violation.

\section{Conclusion}\label{sec:conclusion}

The three roads to $T$-violation each rely on a distinct symmetry principle. The first road, which employs Curie's Principle and the $CPT$ theorem, is by necessity indirect. That's because of the curious fact that Curie's Principle only holds for linear symmetries like $CP$-violation, and not for antilinear symmetries like time reversal. For a more direct test, one can take the second route and apply Kabir's Principle. This restores the possibility of a direct detection of $T$-violation, and indeed has been employed with great success in recent years. For a final test, one can take a third road and apply Wigner's principle. This again allows for a direct test of $T$-violation, which is not contingent on the premises of the $CPT$ theorem, although it requires knowing more about the form of the Hamiltonian.

The way Curie's Principle and Kabir's Principle have been formulated here, it seems at first blush that both routes rely on the assumption of a unitary dynamics. The first approach does so not with Curie's Principle -- it doesn't require unitarity -- but in the application of the $CPT$ theorem. The second approach does so in the application of our formulation of Kabir's Principle. This leads to the appearance that, in extensions of the standard model that relax the assumption of unitarity, we may lose our best existing evidence for $T$-violation. Thus, moving forward, the question of whether $T$-violation will remain an explicit feature of the \emph{fundamental} laws appears, for the moment, to be an open one. This leads immediately to an open question of whether these principles can be formulated in a more general framework, which includes some plausible extensions of the standard model. For an answer, the reader is referred to \citep{ashtekar2013threeroads}.

\bibliographystyle{kp}
\bibliography{/Users/bryanwroberts/Dropbox/ElectronicLibrary/MasterBibliography}
\end{document}